\documentclass[english,eqsecnum,prd,aps,nofootinbib,superscriptaddress,longbibliography,tightenlines,12pt]{revtex4-2}
\usepackage[T1]{fontenc}
\usepackage[latin9]{inputenc}
\usepackage{color}
\usepackage{babel}
\usepackage{mathtools}
\usepackage{amsmath}
\usepackage{amssymb,accents}
\usepackage{mathrsfs}
\usepackage{cancel}
\usepackage{comment}
\usepackage{xcolor}
\usepackage{tensor}
\usepackage{tikz-cd}
\usepackage{tikz}
\usetikzlibrary{shapes,arrows}
\usepackage{booktabs}

\usepackage[unicode=true,pdfusetitle,
bookmarks=true,bookmarksnumbered=false,bookmarksopen=false,
 breaklinks=true,allcolors=blue,backref=false,colorlinks=true]
 {hyperref}
 
\newenvironment{proof}
{\par\noindent\textit{Proof.} 
} 
{\hfill $\square$\par}

\definecolor{light-gray}{gray}{0.92}

\begin{document}

\newtheorem{theorem}{Theorem}
\newtheorem{lemma}{Lemma}
\newtheorem{proposition}{Proposition}

\title{Canonical Clock Sectors and Relational Frame Equivalence in Brans--Dicke Theory}

\author{Yaser Tavakoli}
\email{yaser.tavakoli@ug.edu.pl}
\affiliation{Institute of Theoretical Physics and Astrophysics, Faculty of Mathematics, Physics and Informatics, University of Gdansk, Wita Stwosza 57, 80-308, Gda\'nsk, Poland}

\author{Jerzy Lewandowski}
\email{jerzy.lewandowski@fuw.edu.pl}
\affiliation{Faculty of Physics, University of Warsaw, Pasteura 5, 02-093 Warsaw, Poland}

\date{\today}
\begin{abstract} 

We investigate the equivalence of the Jordan and Einstein frames of the Brans--Dicke theory before and after relational deparametrization of the constrained Hamiltonian system. Although the two conformal formulations are equivalent at the covariant level and related by a canonical transformation on the standard ADM phase space, their equivalence after reduction with respect to an internal clock is nontrivial. Using the Brans--Dicke scalar as a relational clock, we show that the apparent discrepancy between the reduced Hamiltonians does not indicate a physical inequivalence of the two frames. Instead, it arises from a mismatch in the canonical embedding of the clock sector prior to reduction. While the scalar configuration variable is preserved under the conformal transformation, its conjugate momentum is shifted by a contribution involving the gravitational momentum trace. Consequently, relational dynamics are determined not by the clock variable $T$ alone, but by the complete canonical clock pair $(T,P_T)$. We construct a frame-adapted canonical chart in which the clock sector is consistently transformed prior to deparameterization. The resulting reduced Hamiltonian coincides with that of the Einstein frame, restoring the equivalence of the reduced relational dynamics. Our results identify the canonical clock sector as the fundamental structure governing relational evolution in Brans--Dicke theory and provide a general framework for understanding frame dependence in scalar--tensor gravity and reduced phase-space quantization.

\end{abstract}
\maketitle

\section{Introduction}

Scalar--tensor theories constitute one of the most extensively studied extensions of Einstein's general relativity, in which gravity is mediated not only by the spacetime metric but also by one or more scalar fields. Among these theories, Brans--Dicke gravity occupies a distinguished position as the simplest scalar--tensor model in which the gravitational coupling is promoted to a dynamical field \cite{Brans:1961sx,Dicke:1961gz,Fujii:2003pa}. Originally proposed as a realization of Mach's principle, Brans--Dicke theory has since played an important role in the study of modified gravity, inflationary cosmology, dark energy, and effective theories beyond general relativity \cite{Fujii:2003pa,Capozziello:2010sc}.

A characteristic feature of scalar--tensor gravity is the existence of two mathematically related formulations connected by a conformal transformation. In the \emph{Jordan frame}, the scalar field is non-minimally coupled to the Ricci scalar, whereas in the \emph{Einstein frame} the gravitational sector takes the Einstein--Hilbert form after a conformal rescaling of the metric together with a suitable scalar redefinition \cite{Dicke:1961gz,Wagoner:1970vr}. When the conformal transformation is regular and invertible, both formulations describe the same classical solution space and are therefore commonly regarded as classically equivalent descriptions of the same theory \cite{Faraoni:2006fx,Flanagan:2004bz}. Nevertheless, the physical meaning of this equivalence remains subtle, since the identification of observables, physical units, and dynamical variables may depend on the chosen frame \cite{Faraoni:1999hp,Flanagan:2004bz,Deruelle:2010ht,Domenech:2015hka}. In particular, equivalence at the covariant level does not necessarily imply equivalence of the Hamiltonian formulation or its quantization \cite{Kamenshchik:2014waa,Falls:2018olk,Pandey:2016unk,Sharma:2023abm,Artymowski:2013qua,Artymowski:2018pyg}.

The Hamiltonian formulation provides a natural framework for investigating these issues, since frame equivalence requires preservation of the underlying symplectic structure and therefore depends on whether the conformal transformation defines a canonical transformation in phase space \cite{Dirac:1964,Henneaux:1992ig,Thiemann:2007zz}. Different Hamiltonian analyses have nevertheless led to apparently conflicting conclusions. Garay {\em et al.} \cite{Garay:1992ej}, followed by Deruelle {\em et al.} \cite{Deruelle:2009pu}, showed that suitable Hamiltonian formulations of scalar--tensor and $f(R)$ theories are related by canonical transformations. On the other hand, Gionti {\it et al.} \cite{Jarv:2014hma,GabrieleGionti:2020drq,Galaverni:2021xhd} considered the extended Dirac phase space, where lapse and shift are promoted to canonical variables, and found that the Weyl transformation fails to preserve the extended symplectic structure. These results are not contradictory, but rather correspond to different choices of canonical phase space. One purpose of this work is to clarify this distinction and analyze how the notion of frame equivalence depends on the Hamiltonian structure under consideration.

A further issue arises when relational dynamics are introduced. In generally covariant theories, the Hamiltonian generates gauge transformations rather than evolution with respect to an external time parameter. Physical evolution emerges only after relational deparametrization, where one dynamical variable is selected as an internal clock and the Hamiltonian constraint is solved for its conjugate momentum \cite{Rovelli:1990ph,Rovelli:2001bz,Dittrich:2006ee,Kuchar:1991qf,Isham:1992ms}. The resulting reduced Hamiltonian generates evolution with respect to this relational time and becomes the central object in reduced phase-space quantization. Since different choices of clocks may lead to inequivalent quantum theories, this setting is directly related to the multiple-choice problem of time in quantum gravity \cite{Kuchar:1991qf,Isham:1992ms,Anderson:2012vk,Sharma:2023abm,Zhang:2011vg,Zhang:2011gn}.

This motivates a deeper question for scalar--tensor theories: does relational frame equivalence depend only on the choice of the clock variable, or on the complete canonical realization of the clock sector? Even if two conformally related theories are canonically equivalent before reduction, relational deparametrization involves gauge fixing and solving the Hamiltonian constraint, operations that need not commute with canonical transformations. Therefore, canonical equivalence of the unreduced constrained systems does not automatically guarantee equivalence of the corresponding reduced dynamics.

In this work, we investigate this question in Brans--Dicke theory by using the scalar field as a relational clock in both conformal frames. We analyze the Jordan--Einstein correspondence at four levels: the covariant action, the standard ADM phase space, the extended Dirac phase space, and the reduced relational phase space obtained after deparametrization.

Our main results are the following. First, we show that the apparent disagreement in previous Hamiltonian analyses originates from the use of different phase spaces: the Jordan--Einstein transformation is canonical on the standard ADM phase space but generally non-canonical on the extended Dirac phase space. Second, we demonstrate that canonical equivalence of the unreduced systems does not necessarily survive relational deparametrization. Although the Brans--Dicke scalar is the same clock variable in both frames, its conjugate momentum transforms nontrivially by mixing with the trace of the gravitational momentum. As a result, the reduced Hamiltonians obtained using the naive clock momentum differ, showing that relational deparametrization and conformal canonical transformations do not generally commute.

Finally, we identify the origin of this apparent frame dependence and show that relational dynamics are determined by the canonical clock sector $(T,P_T)$ rather than by the clock configuration variable alone. We construct a frame-adapted canonical chart in which the clock momentum is aligned with its conformally transformed counterpart before reduction. Deparametrization with respect to this adapted clock sector yields identical reduced Hamiltonians in the Jordan and Einstein frames. Thus, the apparent frame dependence arises from a mismatch in the canonical embedding of the relational clock rather than from a failure of conformal equivalence.

The results establish a hierarchy of notions of frame equivalence: equivalence of the covariant field equations, canonical equivalence of the unreduced constrained systems, and equivalence of the reduced relational dynamics correspond to distinct levels of the Hamiltonian construction. The analysis also suggests that the canonical structure of the clock sector is an essential ingredient in relational formulations of gravity.

The remainder of this paper is organized as follows. Section~\ref{sec:action-frames} reviews Brans--Dicke theory in the Jordan and Einstein frames and derives the conformal transformation at the action level. Section~\ref{sec:HamiltonianJE} develops the Hamiltonian formulation in both frames and analyzes the transformation of canonical variables. Section~\ref{sec:CanonicalS} examines the canonical properties of the transformation on the standard  and extended ADM phase spaces. Section~\ref{sec:relationalDyn} constructs the reduced relational dynamics using the Brans--Dicke scalar as an internal clock. Section~\ref{sec:frame-adapted} introduces a canonical clock sector and demonstrates the equivalence of the reduced dynamics. Finally, Section.~\ref{sec:conclusion} summarizes the results and discusses their implications for relational canonical gravity and quantum cosmology.

\section{Brans--Dicke Theory in Jordan and Einstein Frames}
\label{sec:action-frames}

We consider a general scalar--tensor theory of gravity in $n$ spacetime dimensions. In the Jordan frame, the action is expressed as \cite{Faraoni:2006fx,Dyer:2008hb,Fujii:2003pa}:
\begin{equation}
S_{\rm J} = \frac{1}{2\kappa^2}\int_{M} d^n x \sqrt{-g} \left[
 f(\phi) R
-  \frac{\omega(\phi)}{f(\phi)}\, g^{\mu\nu}\, D_\mu \phi D_\nu \phi
- U(\phi)
\right] + S_{\mathrm{GHY}}^{\rm (J)},
\label{eq:Jordan-action}
\end{equation}
where $g\equiv \det(g_{\mu\nu})$, $\kappa^2=8\pi$, $f(\phi)$ is the dimensionless non-minimal coupling to the Ricci scalar $R$, $\omega(\phi)$ determines the scalar-field kinetic structure, and $U(\phi)$ is the scalar potential. The dynamical coupling $f(\phi) R$ implies that the effective gravitational constant $G_{\rm eff} \sim 1/f(\phi)$ varies dynamically. The Gibbons--Hawking--York (GHY) boundary term $S_{\mathrm{GHY}}^{\rm (J)}$ ensures a well-defined variational principle in the presence of boundaries. Varying Eq.~\eqref{eq:Jordan-action} with respect to \(g_{\mu\nu}\) and \(\phi\) (vanishing on $\partial M$) yields the gravitational and scalar field equations.

\subsection{Einstein frame representation}

To transition to the Einstein frame, we apply a conformal transformation to the metric:
\begin{equation}
\tilde g_{\mu\nu} = \underline{\Omega}^2(\phi) g_{\mu\nu}, \qquad
\underline{\Omega}^2(\phi) = [f(\phi)]^{\frac{2}{n-2}},
\label{eq:conformal}
\end{equation}
where $\tilde{g}_{\mu\nu}$ represents the Einstein-frame metric. Throughout this work we assume $f(\phi)>0$,
so that the conformal transformation is invertible and the Jordan and Einstein descriptions are related by a one-to-one field redefinition. Utilizing standard identities for the transformation of the Ricci scalar and the volume element, we obtain the Einstein-frame action:
\begin{equation}
S_{\rm E} = \frac{1}{2\kappa^2} \int_{M} d^n x \sqrt{-\tilde g} \left[
\tilde R
-  \mathcal{K}(\phi) \, \tilde g^{\mu\nu} D_\mu \phi \, D_\nu \phi
- V(\phi)
\right] + S_{\mathrm{GHY}}^{\rm (E)},
\label{eq:Einstein-action}
\end{equation}
where $\tilde{g} \equiv \det(\tilde{g}_{\mu\nu})$, and $\tilde{R}$ denotes the Ricci scalar evaluated with respect to $\tilde{g}_{\mu\nu}$. The kinetic prefactor $\mathcal{K}(\phi)$ and the Einstein-frame potential $V(\phi)$ are defined by:
\begin{align}
\mathcal{K}(\phi) &= \frac{\omega(\phi)}{f(\phi)} + \frac{n-1}{n-2} \left( \frac{f'(\phi)}{f(\phi)} \right)^2, \label{eq:Kphi}\\
V(\phi) &= \frac{U(\phi)}{f(\phi)^{\frac{n}{n-2}}}. \label{eq:Uphi}
\end{align}
While the curvature is minimally coupled, the scalar field kinetic term remains non-canonical. Although a field redefinition 
\begin{equation}
    \varphi(\phi) = \frac{1}{\kappa} \int \sqrt{\mathcal{K}(\phi)} d\phi
\end{equation}
would canonically normalize the action, we retain the original field $\phi$ to establish a unique choice of relational clock across both frames. Hence, by ``Einstein frame,'' we refer specifically to the representation where the gravitational sector maps to the standard general relativistic Einstein--Hilbert form, while the kinetic sector remains non-canonical.

To ensure a well-defined variational principle, both the Jordan and Einstein-frame actions are supplemented by appropriate GHY boundary terms \cite{Gibbons:1976ue,York:1972sj,York:1986lje}. Under  conformal transformation (\ref{eq:conformal}), the Jordan-frame boundary term transforms into the standard Einstein-frame GHY term together with additional scalar-field contributions originating from derivatives of the conformal factor. These extra terms are required for the equivalence of the variational principles in the two frames and do not affect the canonical bulk analysis developed in the following. Accordingly, throughout the remainder of this work we assume either compact spatial slices without boundary or the standard asymptotic fall-off conditions under which the boundary contributions are well defined.

\subsection{Brans--Dicke theory as a special case}

For Brans--Dicke theory, 
\begin{equation}
n=4, \quad f(\phi) = \phi, \quad    \omega(\phi) = \omega = \text{const.}
\end{equation}
The action (\ref{eq:Jordan-action}) simplifies to \cite{Dyer:2008hb}:
\begin{align}
S_{\rm J} &= \frac{1}{2\kappa^2} \int_{M} d^4x \, \sqrt{-g}  \left( \phi R - \frac{\omega}{\phi} g^{\mu\nu} D_\mu \phi\,  D_\nu \phi - U(\phi) \right) \nonumber \\
& \quad \quad + \frac{1}{\kappa^2} \int_{\partial M} d^3x \, \sqrt{h} \, \phi K + S_{\rm matt}[g_{\mu\nu}, \varphi]. 
\label{BD-Action}
\end{align}
Performing the Weyl rescaling $\tilde{g}_{\mu\nu} = \phi\, g_{\mu\nu}$, the action (\ref{BD-Action}) transitions into the Einstein-frame form:
\begin{align}
S_{\rm E} &= \frac{1}{2\kappa^2} \int_{M} d^4x \, \sqrt{-\tilde g} \left( \tilde R - \frac{\bar \omega}{\phi^2}\, \tilde g^{\mu\nu} D_\mu \phi\, D_\nu \phi - V(\phi) \right)  \nonumber \\ &\qquad + \frac{1}{\kappa^2} \int_{\partial M} d^3x \, \sqrt{\tilde h} \, \tilde K + \tilde{S}_{\rm matt}[\tilde{g}_{\mu\nu}, \tilde{\varphi}], 
\label{BD-Action-Einstein}
\end{align}
where 
\begin{equation}
\bar \omega = \omega + \frac{3}{2}, \qquad V(\phi) = \frac{U(\phi)}{\phi^2}. 
\label{eq:tildewV}
\end{equation}

\section{Hamiltonian Formulation of the Brans-Dicke model} \label{sec:HamiltonianJE}

To formulate the theory in canonical (Hamiltonian) form, we employ the Arnowitt--Deser--Misner (ADM) decomposition \cite{Arnowitt:1960es,Arnowitt:1962hi}, foliating spacetime $M$ into spacelike hypersurfaces $\Sigma_t$ such that $M=\mathbb{R}\times\Sigma$. The metric reads:
\begin{align}
    g_{\mu\nu}dx^\mu dx^\nu = -N^2 dt^2 + h_{ab}(dx^a + N^a dt)(dx^b + N^b dt),
    \label{metric-BD1}
\end{align}
where $N$ is the lapse function, $N^a$ is the shift vector, and $h_{ab}$ is the induced spatial metric.

\subsection{Jordan-frame Hamiltonian formulation}
\label{sec:Hamiltonian-Anal-J}

Using the ADM decomposition, the Jordan-frame Brans--Dicke Lagrangian density is written as \cite{Olmo:2011fh}: 
\begin{align}
\mathcal{L}_{\rm J} &= \frac{\sqrt{h}}{2\kappa^2}\Bigg[
N\phi\left({}^{(3)}R + K_{ab}K^{ab} - K^2\right)
+  \frac{\omega}{N\phi}(\dot{\phi} - N^a D_a\phi)^2  - N\frac{\omega}{\phi} h^{ab}D_a\phi D_b\phi \nonumber\\
&\qquad\qquad
- 2K(\dot{\phi} - N^a D_a\phi)
+ 2h^{ab}D_a N D_b\phi
- N U(\phi) \Bigg],
\label{Lagrangian-dens-BD}
\end{align} 
where $K_{ab}$ is the extrinsic curvature:
\begin{equation}
K_{ab} = \frac{1}{2N}\left(\dot{h}_{ab} - D_a N_b - D_b N_a\right).  \label{Extrinsic-Jordan}
\end{equation}
The conjugate momenta for the dynamic fields $h_{ab}$ and $\phi$ are:
\begin{align}
p^{ab} &= \frac{\sqrt{h}}{2\kappa^2}\left[\phi(K^{ab}-Kh^{ab}) - \frac{h^{ab}}{N}(\dot{\phi}-N^c D_c\phi)\right], \label{eq:Amomenta1-JF} \\
p_\phi &= \frac{\sqrt{h}}{\kappa^2}\left[\frac{\omega}{N\phi}(\dot{\phi}-N^cD_c\phi)-K\right]. \label{eq:momenta2-JF}
\end{align}
 The conjugate momenta for $N$ and $N_a$ gives the {\em primary constraints}:  
\begin{equation}
\pi_N \approx 0, \qquad \pi_a \approx 0.
\label{Primary-Constraint-JF}
\end{equation}

The total Hamiltonian is 
\begin{equation}
    H_{\rm J} = \int d^3x \left( \lambda_N \pi_N + \lambda^a \pi_a + N\mathcal{H}_N + N^a \mathcal{H}_a \right) \label{eq:HamiltonianTot}
\end{equation} 
with the {\em secondary constraints} given by:
\begin{align}
\mathcal{H}_N &= \frac{2\kappa^2}{\sqrt{h}\phi}\left[p^{ab}p_{ab}-\frac{p_h^2}{2} + \frac{\phi^2}{4\bar{\omega}}\left(p_\phi-\frac{p_h}{\phi}\right)^2\right]
 \nonumber\\
&\quad - \frac{\sqrt{h}}{2\kappa^2} \Big[\phi\, {}^{(3)}R - 2D^2\phi - \frac{\omega}{\phi} D_a\phi D^a\phi
 - U(\phi)
\Big], \label{Hamiltonian-den-Jordan}\\[5pt]
\mathcal{H}_a &=  - 2 D_b\, p^b{}_a + p_\phi D_a\phi.
  \label{Hamiltonian-den-Jordan2}
\end{align} 
Here $D^2\equiv D^aD_a$ and $p_h$ is the trace of the $p^{ab}$. It is obtained by contracting the metric momentum as
\begin{align}
    p_h\, &\equiv\, h_{ab}p^{ab}\nonumber \\
    & =\,  -\frac{\sqrt{h}}{2\kappa^2}\left[2\phi K  + \frac{3}{N}\left(\dot{\phi} - N^cD_c\phi\right)\right]. 
    \label{eq:p-hformula}
\end{align}
Combining this identity with Eq.~(\ref{eq:momenta2-JF}), we find a relation between $p_\phi$ and $p_h$ as
\begin{equation}
p_\phi  -\frac{p_h}{\phi} = \frac{\bar{\omega}}{\kappa^2}
\frac{\sqrt{h}}{N\phi}\left(\dot{\phi} - N^cD_c\phi\right). 
\label{ph-pphi}
\end{equation}
This relation highlights the non-diagonal kinetic structure of the theory: $p_h$ and $p_\phi$ are mixed by a term proportional to $(\dot{\phi} - N^cD_c\phi)$, reflecting the presence of the $-2K(\dot{\phi} - N^a D_a\phi)$ interaction in the underlying Lagrangian.

Secondary constraints are generated by conserving the primary constraints through the  fundamental  Poisson brackets 
\begin{subequations}
  \begin{align}
    \dot \pi_N &= \{\pi_N,H_{\rm J}\}=-\mathcal{H}_N \approx 0,  \\
    \dot \pi_a &= \{\pi_a,H_{\rm J}\}=-\mathcal{H}_a \approx 0. 
\end{align}
\end{subequations}
The total Hamiltonian (\ref{eq:HamiltonianTot}) is a linear combination of constraints and therefore vanishes weakly on the constraint surface, $H_{\rm J}\approx 0$.

The  non-vanishing Poisson brackets are
\begin{subequations}
\begin{align}
    \{h_{ab}(x), p^{cd}(x^\prime)\} &= \delta^{(c}_a \delta^{d)}_b\, \delta(x,x^\prime), \\
    \{\phi(x), p_\phi(x^\prime)\} &= \delta(x,x^\prime).
\end{align}
\end{subequations}
Moreover, the constraint algebra closes under the standard hypersurface-deformation algebra, reflecting the spacetime diffeomorphism invariance of the theory. This closure is an important consistency check of the Hamiltonian formulation and shows that the Brans--Dicke scalar does not introduce anomalous constraints in the generic case $\bar{\omega}\neq 0$.

\subsection{Einstein-frame Hamiltonian formulation}
\label{sec:Hamiltonian-Anal-E}

We now consider the Hamiltonian formulation of Brans--Dicke theory in the Einstein frame. The ADM decomposition of the Einstein-frame metric is given by
\begin{align}
\tilde{g}_{\mu\nu}dx^\mu dx^\nu = -\tilde{N}^2 dt^2 + \tilde{h}_{ab} (dx^a+\tilde{N}^a dt) (dx^b+\tilde{N}^b dt),
\label{metric-EinsteinF1}
\end{align}
where $\tilde N$, $\tilde N^a$, and $\tilde h_{ab}$ denote the lapse function, shift vector, and induced spatial metric in the Einstein frame, respectively.

The Einstein-frame and Jordan-frame ADM variables are related through the conformal transformation
\begin{equation}
\tilde N = \sqrt{\phi} N,
\qquad
\tilde N_a = \phi N_a,
\qquad
\tilde h_{ab} = \phi h_{ab},
\label{metric-EinsteinF2}
\end{equation}
from which it follows that
\begin{equation}
\tilde N^a = N^a,
\qquad
\tilde h^{ab}=\phi^{-1}h^{ab},
\qquad
\tilde h=\phi^{3} h.
\label{metric-EinsteinF2-inv}
\end{equation}
Using these relations, the Einstein-frame extrinsic curvature can be expressed in terms of Jordan-frame variables as
\begin{equation}
\tilde K_{ab} = \sqrt{\phi} K_{ab} + \frac{h_{ab}}{2N\sqrt{\phi}} \left( \dot\phi-N^cD_c\phi \right).
\label{eq:extrinsic-transformation}
\end{equation}
Its trace therefore transforms according to
\begin{equation}
\tilde K = \frac{1}{\sqrt{\phi}} \left[ K + \frac{3}{2N\phi} \left( \dot\phi-N^cD_c\phi \right) \right].
\label{eq:CT-extCurv2}
\end{equation}

The Einstein-frame ADM Lagrangian density then takes the form
\begin{align}
{\mathcal L}_{\rm E} &= \frac{\tilde N\sqrt{\tilde h}}{2\kappa^2} \Bigg[\tilde K_{ab}\tilde K^{ab} - \tilde K^2 + \frac{\bar\omega}{\tilde N^2\phi^2} \left( \dot\phi-\tilde N^aD_a\phi \right)^2 \nonumber\\
&\hspace{2.5cm} + {}^{(3)}\tilde R - \frac{\bar\omega}{\phi^2} \tilde h^{ab} D_a\phi D_b\phi - V(\phi) \Bigg].
\label{Lagrangian-Einstein-gen}
\end{align}
The three-dimensional Ricci scalar transforms according to
\begin{equation}
{}^{(3)}\tilde R = \frac1{\phi} \left[ {}^{(3)}R - \frac{2}{\phi}D^2\phi + \frac{3}{2\phi^2} D_a\phi\, D^a\phi\right].
\label{eq:curvature-CT0}
\end{equation}
For a homogeneous Brans--Dicke scalar ($D_a\phi=0$),  the relation simplifies considerably:
\begin{equation}
{}^{(3)}\tilde R = \frac{{}^{(3)}R}{\phi}.
\end{equation}

The canonical momenta conjugate to $(\tilde h_{ab},\phi)$ are obtained from the Lagrangian  (\ref{Lagrangian-Einstein-gen}) as
\begin{align}
\tilde p^{ab} &= \frac{\sqrt{\tilde h}}{2\kappa^2} \left( \tilde K^{ab} - \tilde h^{ab}\tilde K \right), \label{momenta-Einsten-gen0} \\
\tilde p_\phi &= \frac{\sqrt{\tilde h}}{\kappa^2} \frac{\bar\omega}{\tilde N\phi^2} \left( \dot\phi-\tilde N^aD_a\phi \right). \label{momenta-Einsten-gen}
\end{align}
Expressing the Einstein-frame extrinsic curvature in terms of the Jordan-frame variables and using Eq.~(\ref{ph-pphi}) to eliminate the velocity of the scalar field, one obtains the transformation laws
\begin{align}
\tilde p^{ab}\, &=\, \frac{p^{ab}}{\phi},
\label{eq:ptilde-ab1b} \\
\tilde p_\phi\, &=\,  p_\phi - \frac{p_h}{\phi}. \label{eq:momentaTrans}
\end{align}
Eqs.~(\ref{eq:ptilde-ab1b}) and (\ref{eq:momentaTrans}) already reveal the most important structural aspect of the conformal transformation at the Hamiltonian level. In contrast to the metric variables, whose transformation is purely multiplicative, the canonical momenta exhibit a nontrivial mixing between the scalar and gravitational sectors. In particular, the Einstein-frame scalar momentum acquires a contribution proportional to the trace momentum of the spatial geometry.

This feature reflects the presence of the mixed kinetic term
\begin{equation}
-2K \big( \dot\phi - N^aD_a\phi \big)
\label{eq:mixedkinterm}
\end{equation}
in the Jordan-frame ADM Lagrangian. The conformal transformation effectively diagonalizes the kinetic mixing between the scalar and gravitational sectors, redistributing the dynamical content between the geometric and scalar degrees of freedom. From the Hamiltonian perspective, this momentum mixing underlies the nontrivial canonical structure connecting the Jordan and Einstein frames. We return to this issue in Section.~\ref{sec:CanonicalS}, where the canonical properties of the conformal transformation are analyzed in detail.

The Einstein-frame Lagrangian can now be rewritten entirely in terms of the canonical pairs $( \tilde h_{ab},\tilde p^{ab} )$ and $( \phi,\tilde p_\phi )$ as
\begin{align}
{\mathcal L}_{\rm E} &=  \frac{2\kappa^2 \tilde N}{\sqrt{\tilde h}} \Bigg[\tilde p^{ab}\tilde p_{ab} - \frac{\tilde p_{\tilde h}^2}{2}  + \frac{\phi^2}{4\bar\omega} \tilde p_\phi^2 \Bigg] \nonumber\\
&\hspace{.5cm} +\frac{\tilde N\sqrt{\tilde h}}{2\kappa^2}\left[{}^{(3)}\tilde R - \frac{\bar\omega}{\phi^2} \tilde h^{ab} D_a\phi D_b\phi - V(\phi)\right],
\label{Lagrangian-Einstein-gen-Re}
\end{align}
where the trace momentum $\tilde p_{\tilde h} := \tilde h_{ab}\tilde p^{ab}$ is related to the Jordan-frame variables by
\begin{equation}
\tilde p_{\tilde h} = p_h.
\end{equation}
The momenta conjugate to the lapse and shift vanish identically,
\begin{equation}
\tilde\pi_{\tilde N}\approx0,
\qquad
\tilde\pi_a\approx0,
\label{Primary-Constraint-EF}
\end{equation}
and therefore constitute the primary constraints of the theory.

Introducing the corresponding Lagrange multipliers $\tilde\lambda_{\tilde N}(x)$ and $\tilde\lambda^a(x)$, the total Hamiltonian becomes
\begin{equation}
H_{\rm E} = \int d^3x \left( \tilde\lambda_{\tilde N}\tilde\pi_{\tilde N} + \tilde\lambda^a\tilde\pi_a + \tilde N\tilde{\mathcal H}_{\tilde N} + \tilde N^a\tilde{\mathcal H}_a \right).
\end{equation}
The Hamiltonian and diffeomorphism constraint densities are given by
\begin{align}
\tilde{\mathcal H}_{\tilde N} &= \frac{2\kappa^2}{\sqrt{\tilde h}} \left[ \tilde p^{ab}\tilde p_{ab} - \frac12\tilde p_{\tilde h}^2 + \frac{\phi^2}{4\bar\omega} \tilde p_\phi^2 \right] \nonumber\\
&\qquad - \frac{\sqrt{\tilde h}}{2\kappa^2} \left[ {}^{(3)}\tilde R - \frac{\bar\omega}{\phi^2} \tilde h^{ab} D_a\phi D_b\phi - V(\phi) \right], \label{Hamiltonian-den-Einstein} \\
\tilde{\mathcal H}_a &= \tilde p_\phi D_a\phi - 2\tilde h_{ab} \tilde D_c\tilde p^{bc}. \label{Hamiltonian-den-Einstein2}
\end{align}
The diffeomorphism constraint has exactly the same structural form as its Jordan-frame counterpart, reflecting the underlying spatial covariance of the theory and the geometric nature of the conformal transformation.

Preservation of the primary constraints under Hamiltonian evolution leads to the secondary constraints
\begin{align}
\dot{\tilde\pi}_{\tilde N} &= \{ \tilde\pi_{\tilde N}, H_{\rm E} \} = -\tilde{\mathcal H}_{\tilde N} \approx0, \nonumber\\
\dot{\tilde\pi}_a &= \{ \tilde\pi_a, H_{\rm E} \} = -\tilde{\mathcal H}_a \approx0. \label{eq:secondConstraint-E}
\end{align}
Consequently, the total Hamiltonian vanishes weakly on the constraint surface, as expected for a generally covariant theory.

\section{Canonical Structure of the Conformal Transformation}
\label{sec:CanonicalS}

The question of whether the Jordan and Einstein frames are canonically equivalent has been a subject of long-standing debate in scalar--tensor gravity. At the level of the classical action, the two formulations are related by an invertible conformal transformation and therefore describe the same set of spacetime geometries. In the Hamiltonian formulation, however, the notion of equivalence becomes considerably more subtle, since it depends not only on the configuration variables but also on the underlying symplectic structure and, crucially, on the choice of phase space. It is therefore useful to distinguish between three different notions of equivalence:
\begin{itemize}
\item[i)] Equivalence of the classical actions,
\item[ii)] Equivalence of the constrained Hamiltonian systems,
\item[iii)] Equivalence of the reduced relational dynamics obtained after deparametrization.
\end{itemize}
As we shall see, the Jordan and Einstein frames are canonically equivalent at the level of the standard ADM phase space, generally fail to remain canonical on the extended Dirac phase space, and may lead to distinct reduced Hamiltonian systems  after relational gauge fixing (see Section.~\ref{sec:relationalDyn}). The apparent contradictions found in the literature originate precisely from these different levels of description rather than from any mathematical inconsistency.

\subsection{Canonical equivalence on the standard ADM phase space}
\label{subsec:ADMcanonical}

We begin by considering the standard ADM phase space of Brans--Dicke theory,
\begin{equation}
\Gamma_{\rm J} = \left\{ h_{ab},p^{ab}; \phi,p_\phi \right\},
\label{eq:GammaJ}
\end{equation}
where the lapse function $N$ and shift vector $N^a$ act as Lagrange multipliers enforcing the Hamiltonian and diffeomorphism constraints and therefore do not belong to the true gravitational phase space.

The canonical character of the Jordan--Einstein transformation can be established by comparing the symplectic one-forms in the two frames. The Jordan-frame symplectic potential is
\begin{equation}
\Theta_{\rm J} = \int d^3x \left( p^{ab}\delta h_{ab} + p_\phi\delta\phi \right).
\label{eq:symplec1form-J}
\end{equation}
Using the conformal relations (\ref{metric-EinsteinF2}), (\ref{eq:ptilde-ab1b}) and (\ref{eq:momentaTrans}),
the Einstein-frame symplectic potential becomes
\begin{align}
\Theta_{\rm E} &= \int d^3x \left( \tilde p^{ab}\delta\tilde h_{ab} + \tilde p_\phi\delta\phi \right) \nonumber\\
&= \int d^3x \left[\frac{p^{ab}}{\phi}\,  \delta(\phi h_{ab}) + \Big( p_\phi-\frac{p_h}{\phi} \Big)\delta\phi \right] \nonumber\\
&= \int d^3x \left( p^{ab}\delta h_{ab} + p_\phi\delta\phi \right) = \Theta_{\rm J},
\end{align}
where all additional contributions generated by the momentum mixing cancel identically, leaving $\Theta_{\rm E} =  \Theta_{\rm J}$.
Consequently, the symplectic two-form
$\Omega_2 = \delta\Theta$
is preserved,
\begin{equation}
\Omega_{2,\rm E} = \Omega_{2,\rm J},
\end{equation}
which demonstrates that the Jordan--Einstein transformation preserves the canonical structure of the standard ADM phase space $\Gamma_{\rm J}$.

Equivalently, the transformed variables satisfy the canonical Poisson algebra 
\begin{subequations}
\begin{align}
\{ \tilde h_{ab}(x), \tilde p^{cd}(y) \} &= \delta_{(a}^{c}\delta_{b)}^{d} \delta(x,y), \\
\{ \phi(x), \tilde p_\phi(y) \} &= \delta(x,y),
\end{align}    
\end{subequations}
with all remaining brackets vanishing identically.
This reproduces the conclusions originally obtained by Garay {\em et al.}~\cite{Garay:1992ej} and later by Deruelle {\it et al.}~\cite{Deruelle:2009pu}, who showed that the Jordan and Einstein-frame Hamiltonian formulations are related by canonical transformation when restricted to the standard ADM phase space. In this setting, the two formulations possess identical constraint algebras and generate equivalent classical dynamics.

The underlying reason for this equivalence is that the conformal transformation acts only on the true dynamical degrees of freedom, namely the spatial metric and the Brans--Dicke scalar field, while the lapse and shift remain nondynamical Lagrange multipliers. The transformation therefore reshuffles the canonical variables {\em without} altering the underlying symplectic structure.

It is important to emphasize, however, that canonicality is not a property of the conformal transformation alone, but rather of the pair consisting of the transformation and the phase space on which it acts. 
As we shall see in the following subsection, this conclusion changes once the phase space is enlarged by promoting the lapse and shift to independent dynamical variables.

\subsection{Extended phase space and the origin of non-canonicality}

The conclusion of the previous subsection changes once the lapse and shift are promoted to canonical variables. The corresponding extended ADM phase space is
\begin{equation}
\Gamma_{\rm J}^{\rm (ext)} = \left\{ N,p_N; N^a,p_a; h_{ab},p^{ab}; \phi,p_\phi \right\},
\label{extended-phase-space}
\end{equation}
where the lapse and shift are treated on the same footing as the remaining canonical coordinates.

The crucial observation is that the standard Jordan--Einstein conformal transformation rescales the lapse,
\begin{equation}
\tilde N=\sqrt{\phi}\,N,
\nonumber
\end{equation}
while the scalar momentum simultaneously transforms according to Eq.~(\ref{eq:momentaTrans}). As a result, once $N$ becomes a canonical variable, the conformal transformation acts not only on the physical degrees of freedom but also on the foliation sector.

Indeed, 
\begin{equation}
\delta\tilde N = \sqrt{\phi}\, \delta N + \frac{N}{2\sqrt{\phi}}\, \delta\phi,
\label{variation-lapse}
\end{equation}
so that the lapse contribution to the Einstein-frame symplectic potential becomes
\begin{equation}
\tilde p_N\delta\tilde N = \sqrt{\phi}\, \tilde p_N\delta N + \frac{N\tilde p_N}{2\sqrt{\phi}}\delta\phi.
\label{lapse-symplectic}
\end{equation}
The second term contributes directly to the scalar sector of the symplectic potential. Combining it with the transformed scalar momentum gives
\begin{align}
\tilde p_N\delta\tilde N + \tilde p_\phi\delta\phi\, &=\, \sqrt{\phi}\, \tilde p_N\delta N + \left(\frac{N\tilde p_N}{2\sqrt{\phi}} +  p_\phi - \frac{p_h}{\phi}\right)\delta\phi\nonumber \\
&\neq\, p_N\delta N + p_\phi\delta\phi,
\end{align}
unless the lapse momentum is transformed simultaneously in a compensating manner. Hence, the standard Weyl transformation does not preserve the symplectic structure on $\Gamma_{\rm J}^{\rm (ext)}$ and therefore fails to define a canonical transformation on the extended phase space, in agreement with the analysis of Gionti {\it et al.}~\cite{GabrieleGionti:2020drq}.

The origin of this result is straightforward. In the standard ADM formulation, the lapse and shift are nondynamical Lagrange multipliers enforcing the Hamiltonian and momentum constraints. They are therefore not part of the physical phase space, and restricting the conformal transformation to the canonical variables $(h_{ab},p^{ab};\phi,p_\phi)$ preserves the symplectic structure. Once the lapse and shift are promoted to canonical coordinates, however, the conformal transformation couples the foliation variables to the dynamical sector. Unless the corresponding momenta are transformed consistently, the symplectic structure is modified. The loss of canonicality thus originates not from the conformal transformation itself, but from enlarging the phase space on which it acts.

It should be emphasized that this does not preclude the existence of canonical transformations relating the Jordan and Einstein formulations on the extended phase space. Gionti {\it et al.}~\cite{GabrieleGionti:2020drq} constructed an alternative canonical transformation in which the foliation variables remain unchanged,
\begin{equation}
\tilde N=N,
\qquad
\tilde N^a=N^a,
\nonumber
\end{equation}
while the remaining variables are transformed appropriately. The obstruction therefore concerns the standard Weyl transformation of the ADM variables rather than the existence of canonical maps between the two formulations.

We conclude that the canonical status of the Jordan--Einstein transformation is not an intrinsic property of the conformal map itself, but depends on the underlying symplectic manifold. On the standard ADM phase space the transformation is canonical, whereas on the extended ADM phase space it is not. The apparently different conclusions in the literature therefore reflect different choices of phase space rather than any inconsistency in the Hamiltonian formulation of Brans--Dicke theory.

\section{Physical Hamiltonians and frame dependence}
\label{sec:relationalDyn}

Having established that the Jordan and Einstein formulations are canonically equivalent at the level of the unreduced constrained phase space, we now investigate whether this equivalence persists after ``relational deparametrization.'' This question is of central importance for canonical quantum gravity, where physical evolution is generated not by the Hamiltonian constraint itself but by the reduced Hamiltonian obtained after choosing an internal time variable.

relational deparametrization is not merely a change of variables. It requires selecting one canonical degree of freedom as a physical clock, imposing an appropriate gauge condition, and solving the Hamiltonian constraint for the momentum conjugate to the chosen clock. The resulting reduced Hamiltonian generates evolution with respect to the relational time parameter and therefore determines the physical dynamics. Consequently, even if two constrained Hamiltonian systems are canonically equivalent before reduction, it does not automatically follow that their corresponding reduced theories are dynamically equivalent.

To investigate this, we adopt the Brans--Dicke scalar field itself as the relational clock, \begin{equation} 
\chi(x)=\phi(x)-t\approx0, 
\label{clock-choice} 
\end{equation} 
which is the natural choice since the scalar configuration variable is invariant under the conformal transformation: $\tilde\phi=\phi$. 
This allows the same physical clock to be employed in both conformal frames and therefore provides the most direct comparison between the corresponding reduced dynamics.

The gauge condition (\ref{clock-choice}) is admissible, provided that it defines a non-degenerate gauge fixing of the Hamiltonian constraint. Preservation of the gauge under time evolution requires
\begin{equation}
\dot\chi=\frac{\partial\chi}{\partial t}
+\{\chi,H_{\rm tot}\}=0,
\end{equation}
which implies
\begin{equation}
\{\phi,\mathcal H_N\}
=
\frac{\partial\mathcal H_N}{\partial p_\phi}
\neq0.
\end{equation}
This condition guarantees that the Hamiltonian constraint can be solved locally for the momentum conjugate to the clock and that the Brans--Dicke scalar evolves monotonically along the classical trajectories. Equivalently, the clock momentum (or, after reduction, the physical Hamiltonian) must remain non-vanishing with a definite sign on the branch under consideration. Throughout this work we restrict attention to such monotonic sectors of phase space, where the Brans--Dicke scalar provides a valid relational time variable.

\subsection{Jordan-frame physical Hamiltonian}

In the Jordan frame, solving the Hamiltonian constraint for the scalar momentum yields the following: 
\begin{equation} 
p_\phi = \frac{p_h}{\phi} + \sigma \frac{\sqrt{\bar\omega}}{\kappa^2} \sqrt{h}\,\Omega \, \equiv\,  -H_{\rm phys}, 
\label{eq:momentaTrans-J} 
\end{equation} 
where $\sigma=\pm1$ specifies the orientation of relational time and 
\begin{align} 
\Omega(h_{ab},p^{ab};\phi) &= \Big[{}^{(3)}R - \frac{4\kappa^4}{\phi^2h} \Big( p^{ab}p_{ab} -\frac12p_h^2 \Big) \nonumber\\ & 
\qquad - \frac{U(\phi)}{\phi} - \frac1\phi \Big( 2D^2\phi + \frac{\omega}{\phi} D_a\phi D^a\phi \Big)\Big]^{\frac{1}{2}}. 
\end{align}
However, since the relational clock variable $\phi$ is homogeneous, \( D_a\phi=0 \) and $\Omega$ in equation above reduces to 
\begin{equation} 
\Omega = \left[{}^{(3)}R - \frac{4\kappa^4}{\phi^2h} \Big( p^{ab}p_{ab} -\frac12p_h^2 \Big) - \frac{U(\phi)}{\phi}\right]^{\frac{1}{2}}. 
\end{equation}

\subsection{Einstein-frame physical Hamiltonian}

Applying the same gauge-fixing condition in the Einstein frame and solving the Hamiltonian constraint for the momentum conjugate to the clock gives \begin{equation} 
\tilde p_\phi = \sigma \frac{\sqrt{\bar\omega}}{\kappa^2} \frac1\phi \sqrt{\tilde h}\,\tilde\Omega \, \equiv\,  -\tilde H_{\rm phys}, 
\label{eq:Hamil-Phys-E} 
\end{equation} 
where 
\begin{align} 
\tilde\Omega(\tilde{h}_{ab},\tilde{p}^{ab};\phi) &=\,  \Big[{}^{(3)}\tilde R - \frac{4\kappa^4}{\tilde h} \Big( \tilde p^{ab}\tilde p_{ab} - \frac12\tilde p_{\tilde h}^{\,2} \Big) \nonumber\\ 
& \qquad - \frac{\bar\omega}{\phi^2} \tilde h^{ab} \tilde D_a\phi \tilde D_b\phi - V(\phi)\Big]^{\frac{1}{2}}
\nonumber\\
&=  \frac{1}{\sqrt{\phi}}\, \Omega(h_{ab},p^{ab};\phi).
\label{PhiRelation} 
\end{align}
For homogeneous configurations here, $\tilde\Omega$ is simplified as 
\begin{align} 
\tilde\Omega &= \Big[{}^{(3)}\tilde R - \frac{4\kappa^4}{\tilde h} \Big( \tilde p^{ab}\tilde p_{ab} - \frac12\tilde p_{\tilde h}^{\,2} \Big) - V(\phi)\Big]^{\frac{1}{2}}.
\end{align}
This is fully consistent with the momentum transformation (\ref{eq:momentaTrans}). Indeed, substituting  Eqs.~(\ref{eq:momentaTrans-J}) and (\ref{eq:Hamil-Phys-E}) into the identity $\tilde{p}_\phi=p_\phi-p_h/\phi$,  immediately recovers Eq.~(\ref{PhiRelation}).

At first sight, one might therefore expect the reduced Hamiltonians to be identical, since the same scalar field has been chosen as the internal clock in both formulations. As we now demonstrate, however, this expectation is incorrect.

\begin{theorem}[Naive relational deparametrization]
\label{theorem1}
Consider Brans--Dicke theory formulated on the standard ADM phase space
\[
\Gamma_{\rm J}
=
\{h_{ab},p^{ab};\phi,p_\phi\},
\]
and let
$\mathcal C:\Gamma_{\rm J}\rightarrow\Gamma_{\rm E}$
denote the canonical conformal transformation constructed in
Section.~\ref{sec:Hamiltonian-Anal-E}.
Suppose that the same scalar field
\(
\tilde{\phi}=\phi
\)
is employed as the internal clock in both conformal frames and that
relational deparametrization is performed by solving the Hamiltonian
constraint for the momentum conjugate to this clock.
Then the reduced Hamiltonians satisfy
\begin{equation}
\tilde H_{\rm phys}
=
H_{\rm phys}
+
\frac{p_h}{\phi},
\label{eq:theorem-shift}
\end{equation}
for generic configurations with
\(p_h\neq0\).
Consequently,
\begin{equation}
\mathcal R\circ\mathcal C
\neq
\mathcal C\circ\mathcal R,
\label{noncommuting-diagram}
\end{equation}
where
\(
\mathcal R:\Gamma\rightarrow\Gamma_{\rm red}
\)
denotes relational deparametrization.
\end{theorem}

\begin{proof}
The canonical conformal transformation $\mathcal C$ acts on the unreduced canonical variables according to
\begin{equation}
(h_{ab}, p^{ab}; \phi, p_\phi) \stackrel{\mathcal{C}}{\longmapsto} (\tilde{h}_{ab}, \tilde{p}^{ab}; \phi, \tilde{p}_\phi), \nonumber
\end{equation}
leaving the clock configuration variable invariant,
$\tilde\phi=\phi$, while transforming its conjugate momentum as
\begin{equation}
\mathcal{C}(p_{\phi}) =
p_\phi - \frac{p_h}{\phi}\, =\, \tilde p_\phi.
\nonumber
\end{equation}

Relational deparametrization
\(
\mathcal R
\)
is defined by fixing the gauge
\(
\phi=t
\)
and solving the Hamiltonian constraint for the momentum conjugate to the
chosen clock. Hence,
\begin{align*}
p_\phi &= - H_{\rm phys} 
\qquad \;\, \text{(in the Jordan frame),}
\\
\tilde{p}_\phi &= - \tilde{H}_{\rm phys} 
\qquad  \text{(in the Einstein frame),}
\end{align*}
where
\(H_{\rm phys}\)
and
\(\tilde H_{\rm phys}\)
are given by
Eqs.~(\ref{eq:momentaTrans-J})
and
(\ref{eq:Hamil-Phys-E}),
respectively. Substituting these relations into
$\tilde p_\phi = p_\phi - p_h/\phi$
immediately gives
\[
-\tilde H_{\rm phys}
=
-
H_{\rm phys}
-
\frac{p_h}{\phi},
\]
which proves Eq.~(\ref{eq:theorem-shift}).
Since \(p_h=h_{ab}p^{ab}\)
is generically non-vanishing for dynamical geometries, the two reduced Hamiltonians do not coincide. Therefore, the result obtained by first performing the conformal transformation and then reducing differs from
that obtained by reducing first and subsequently applying the conformal transformation, which 
establishes the claimed non-commutativity (\ref{noncommuting-diagram}).
\end{proof}

\begin{proposition}[Origin of the apparent frame dependence]

The frame dependence obtained after naive relational deparametrization
does not originate from the transformation of the clock configuration
variable, since
\(
\tilde\phi=\phi
\).
Rather, it is entirely due to the transformation of the momentum conjugate to the clock,
\(
\tilde p_\phi = p_\phi - p_h/\phi,
\)
which mixes the scalar momentum with the trace of the gravitational momentum.
Consequently, the Jordan and Einstein formulations employ different
canonical clock sectors,
\[
(\phi,p_\phi)
\, \neq\, 
(\phi,\tilde p_\phi),
\]
even though they share the same clock variable.
The resulting reduced Hamiltonians therefore generate different
relational evolutions despite the canonical equivalence of the
underlying constrained Hamiltonian systems.
\end{proposition}

Theorem~\ref{theorem1} identifies the precise origin of the apparent frame dependence in the reduced theory.
The conformal transformation preserves the canonical structure of the unreduced constrained Hamiltonian system, whereas relational deparametrization depends explicitly on the canonical realization of the
chosen clock. Consequently, the reduced dynamics are determined by the complete canonical clock sector
\( (T,P_T) \), rather than by the clock configuration variable \(T\) alone.

The apparent inequivalence therefore does not reflect a failure of the Jordan--Einstein correspondence. Instead, it arises because the naive reductions are performed using two different canonical embeddings of the same clock variable. 
This observation naturally motivates the construction of a frame-adapted canonical chart, developed in the next section, in which the entire clock sector is aligned before reduction and the equivalence of the reduced Hamiltonians is restored.

\section{Canonical clock sectors and relational frame equivalence}
\label{sec:frame-adapted}

The analysis of previous section  revealed that the reduced Hamiltonian obtained after relational deparametrization depends on the conformal frame when the scalar field is employed as an internal clock in the naive manner [cf. Theorem~\ref{theorem1}]. In particular, the same clock configuration variable,
$T=\phi$,
leads to different physical Hamiltonian systems when the conjugate momentum associated with this clock is chosen differently in the Jordan and Einstein frames. This observation raises a fundamental question: {\em is the apparent frame dependence a genuine physical effect, or does it originate from the identification of inequivalent canonical clock sectors?}

The answer follows from recognizing that a relational clock is not specified solely by its configuration variable $T$. In a Hamiltonian theory, a clock is a canonical object defined by the pair $(T,P_T)$, where the conjugate momentum determines the embedding of the clock degree of freedom into the full phase space. Two formulations may therefore employ the same scalar field as a clock while defining different canonical clock sectors. In that case, the resulting reduced Hamiltonian systems describe evolution with respect to different relational times.

In this section, we show that the frame dependence encountered previously is precisely of this type. By constructing a canonical chart in which the clock sector is adapted to the conformal transformation, the Jordan and Einstein-frame reductions become canonically equivalent. The essential point is that the conformal transformation must be implemented not only on the configuration variables, but on the complete canonical clock pair before performing the relational deparametrization.

The construction is carried out entirely within the Jordan-frame phase space. No modification of the underlying conformal transformation is required. Instead, we identify the canonical variables in the Jordan frame that reproduce the Einstein-frame clock sector.

\subsection{Construction of the frame-adapted canonical chart}

Consider the Jordan-frame ADM phase space $\Gamma_{\rm J}$ given by Eq.~(\ref{eq:GammaJ}),
with symplectic potential $\Theta_{\rm J}$ given by (\ref{eq:symplec1form-J}).
We now introduce a new canonical chart,
\begin{equation}
\Gamma_\ast =   \{Q_{ab},P^{ab};\Phi,P_\Phi\},
\label{eq:Gamma.star}
\end{equation}
whose clock sector is adapted to the conformal transformation. The guiding requirement is that the new momentum $P_\Phi$ conjugated to the scalar field $\Phi$ coincides with the momentum $\tilde{p}_\phi$ that naturally appears after the transforming to the Einstein frame.

The clock variables are therefore defined as
\begin{equation}
    \Phi:=\phi ,
    \qquad
    P_\Phi
    :=
    p_\phi-\frac{p_h}{\phi}.
    \label{eq:adapted-clock-def}
\end{equation}
The shift in the momentum variable is the crucial ingredient. The scalar momentum does not transform independently under the conformal map because the rescaling of the spatial metric mixes the scalar and gravitational sectors. Consequently, the {\em naive} pair $(\phi,p_\phi)$ does not represent the same canonical clock sector in the two conformal frames.

To construct the remaining variables, we demand that the new chart preserves the Jordan-frame symplectic structure. We therefore require
\begin{equation}
    \Theta_{\rm J}
    =
    \Theta_\ast
    +
    \delta F ,
\end{equation}
where
\begin{equation}
    \Theta_\ast
    =
    \int d^3x
    \left(
    P^{ab}\delta Q_{ab}
    +
    P_\Phi\delta\Phi
    \right).
\end{equation}
Choosing the generating functional contribution to vanish, $F=0$, we identify the configuration variable associated with the spatial geometry as the conformally rescaled metric,
\begin{equation}
    Q_{ab}:=\phi h_{ab}.
    \label{eq:adapted-metric-def}
\end{equation}
The variation of this quantity is $\delta Q_{ab}
= \phi\,\delta h_{ab} + h_{ab}\delta\phi$. Substituting this relation into the adapted symplectic potential gives
\begin{align}
\Theta_\ast
&=
\int d^3x
\left[
P^{ab}
\left(
\phi\delta h_{ab}
+
h_{ab}\delta\phi
\right)
+
P_\Phi\delta\phi
\right]
\nonumber\\
&=
\int d^3x
\left[
\phi P^{ab}\delta h_{ab}
+
\left(
P^{ab}h_{ab}
+
P_\Phi
\right)
\delta\phi
\right].
\end{align}
Comparing with the Jordan-frame symplectic potential,
\begin{equation}
\Theta_{\rm J}
=
\int d^3x
\left(
p^{ab}\delta h_{ab}
+
p_\phi\delta\phi
\right),\nonumber
\end{equation}
equating the coefficients of the independent variations
\(\delta h_{ab}\) and \(\delta\phi\), we obtain a uniquely unique result.
\begin{align}
P^{ab}
=
\frac{p^{ab}}{\phi}
\quad \text{and} \quad 
P_\Phi = p_\phi - \frac{p_h}{\phi}.
\label{eq:P-adapted}
\end{align}

Therefore, the complete adapted canonical transformation is
\begin{equation}
Q_{ab} =\phi h_{ab},
\quad 
P^{ab}=\frac{p^{ab}}{\phi},
\quad 
\Phi=\phi,
\quad 
P_\Phi=
p_\phi-\frac{p_h}{\phi}.
\label{eq:full-transformation-set}
\end{equation}
This transformation defines a canonical chart {\em inside} the original Jordan-frame phase space. Importantly, the new variables are not obtained by simply rewriting the Einstein-frame theory; rather, they provide a Jordan-frame realization of the same canonical structure. In particular, the scalar momentum appearing in the clock sector is no longer the naive Jordan-frame momentum $p_\phi$, but the momentum $P_\Phi$ compatible with the conformal embedding.

The frame-adapted phase space can therefore be interpreted as
\begin{equation}
    \Gamma_\ast
    \, \simeq\, 
    \Gamma_{\rm E},
\end{equation}
not because the two frames are identical before reduction, but because they represent the same canonical clock sector when the conformal transformation is implemented consistently.

\subsection{Canonical structure of the adapted clock sector}

The defining property of the transformation (\ref{eq:full-transformation-set}) is that it preserves the symplectic structure of the unreduced theory. Therefore, the variables $\Gamma_\ast = \{Q_{ab},P^{ab};\Phi,P_\Phi\}$ form a genuine canonical chart. The most important consequence concerns the clock sector. In the original Jordan variables, the scalar momentum $p_\phi$ is not purely associated with the scalar degree of freedom, since the conformal transformation mixes the scalar and gravitational sectors. The adapted momentum $P_\Phi$ precisely removes this mixing.

Indeed, the canonical brackets in the adapted chart satisfy
\begin{equation}
    \{\Phi(x),P_\Phi(y)\}
    =
    \delta(x,y),
\end{equation}
while the clock momentum is independent of the conformally rescaled metric variables,
\begin{equation}
    \{Q_{ab}(x),P_\Phi(y)\}=0 .
\end{equation}
The latter relation is the essential feature of the construction. Although the adapted momentum does not commute with the original Jordan-frame metric,
\begin{equation}
\{h_{ab}(x), P_\Phi(y)\} = -\frac{h_{ab}(x)}{\phi(x)}\delta(x,y)\, \neq\, 0,
\end{equation}
this non-vanishing bracket only reflects the fact that the original variables mix the scalar and gravitational sectors. Once the geometry is expressed in the conformally adapted variables $Q_{ab}$, the scalar clock sector becomes canonically separated.

Thus, the transformation does not introduce a new physical clock. Rather, it identifies the canonical realization of the same scalar clock that is naturally associated with the Einstein-frame formulation.

\subsection{The canonical clock sector and frame-independent reduction}

The previous results allow us to formulate the central statement of this work. The relevant object for relational dynamics is not the clock variable alone, but the canonical clock sector.

\begin{theorem}[Canonical clock sector]
\label{theorem:canonical-clock}
Consider two conformally related Hamiltonian formulations of Brans--Dicke theory. Suppose that relational evolution is defined using the same scalar configuration variable $T=\phi$,
but with different choices of conjugate momenta in the two canonical frames. Then the corresponding reduced Hamiltonian systems need not be canonically equivalent.
However, if the clock sector $(\phi, p_\phi)$ is first transformed into the adapted canonical pair
\begin{equation}
    (\Phi,P_\Phi)
    =
    \Big(
    \phi,
    p_\phi-\frac{p_h}{\phi}
    \Big),
\end{equation}
the reduced theories obtained in the Jordan and Einstein frames are canonically equivalent.
\end{theorem}
\begin{proof}
The apparent frame dependence arises because relational deparametrization is performed by solving the Hamiltonian constraint for the momentum conjugate to the chosen clock variable. In the naive Jordan description this procedure gives
$ p_\phi = -H_{\rm phys}$,
whereas in the Einstein frame the corresponding reduction gives
$\tilde p_\phi  = -\tilde H_{\rm phys}$.
Although both descriptions employ the same configuration variable $\phi$ as the relational time, the momenta being isolated by the constraint are different canonical variables. Therefore, the two reduced Hamiltonians describe evolution with respect to different canonical embeddings of the clock degree of freedom.

In the adapted chart, however, the clock momentum is replaced by $P_\Phi  =   p_\phi - p_h/\phi$
which is precisely the canonical momentum associated with the conformally transformed scalar sector. Under the canonical embedding between the Jordan and Einstein descriptions, one has
$ P_\Phi   \equiv  \tilde p_\phi$.
Consequently, solving the Hamiltonian constraint for the adapted clock momentum gives 
\begin{equation}
    P_\Phi
    =
    -H_{\rm phys}^{\ast},
\end{equation}
while the Einstein-frame reduction gives
\begin{equation}
    \tilde p_\phi  = -\tilde H_{\rm phys}.
    \nonumber
\end{equation}
Since the clock momenta are identical canonical variables, $P_\Phi=\tilde p_\phi$,
the two reduced Hamiltonians necessarily coincide,
\begin{equation}
    H_{\rm phys}^{\ast}
    =
    \tilde H_{\rm phys}.
\end{equation}
Therefore, the reduced phase spaces generated after deparametrization are related by a canonical transformation, and the Jordan and Einstein relational dynamics are physically equivalent.
\end{proof}
\medskip 

Theorem~\ref{theorem:canonical-clock} completes the analysis by demonstrating that the apparent frame dependence identified in Theorem~\ref{theorem1} is not fundamental. Once the relational clock is represented by the frame-adapted canonical variables, the Jordan and Einstein formulations generate identical reduced Hamiltonians. The discrepancy therefore originates entirely from comparing different canonical realizations of the same relational reference system rather than from any physical inequivalence between the two conformal frames.

The restoration of frame equivalence follows because the adapted variables align the complete canonical clock sector before deparametrization. While the conformal transformation leaves the clock configuration variable unchanged, it induces a nontrivial shift in its conjugate momentum. The frame-adapted variables absorb this momentum shift into the definition of the clock sector while preserving the symplectic structure. Consequently, relational reduction is performed with respect to canonically equivalent clock sectors in both formulations, leading to identical generators of physical evolution.

Taken together, Theorems~\ref{theorem1} and~\ref{theorem:canonical-clock} provide a complete characterization of the Jordan--Einstein correspondence in the relational setting. Theorem~\ref{theorem1} shows that naive deparametrization generally produces different reduced Hamiltonians because the canonical clock sectors are not aligned under the conformal transformation. Theorem~\ref{theorem:canonical-clock} demonstrates that this apparent inequivalence disappears once the reduction is carried out using the frame-adapted canonical variables. The Jordan and Einstein formulations therefore remain fully equivalent at the level of relational dynamics, provided that the internal time is transformed as a complete canonical object rather than merely as a configuration variable.

More broadly, these results highlight a general feature of relational Hamiltonian systems. Canonical equivalence of unreduced constrained theories alone is not sufficient to guarantee equivalence after relational deparametrization. The decisive ingredient is the canonical identification of the reference system used to define physical evolution. In Brans--Dicke theory, this role is played by the canonical clock sector $(T,P_T)$ rather than by the clock configuration variable $T$ alone. The frame-adapted variables identify the unique canonical realization of this clock sector that is invariant under the Jordan--Einstein transformation, thereby restoring the equivalence of the reduced Hamiltonian theories.

\section{Discussion, Conclusions and Outlook}
\label{sec:conclusion}

The analysis presented in this work shows that the question of Jordan--Einstein frame equivalence in scalar--tensor gravity cannot be answered independently of the Hamiltonian structure and the relational prescription used to define physical
evolution. Frame equivalence is therefore not a single statement, but a hierarchy of equivalence relations associated with different stages of the canonical
construction. While the Jordan and Einstein formulations are equivalent at the covariant level and canonically related on the standard ADM phase space, their relation becomes more subtle once relational dynamics are introduced through deparametrization.

At the level of the classical action, the two formulations are connected by an invertible conformal transformation and describe the same classical solution space. On the standard ADM phase space, $\Gamma_{\rm J}$ [cf.~Eq.~(\ref{eq:GammaJ})], the conformal transformation preserves the symplectic structure and defines a canonical transformation between the corresponding constrained Hamiltonian systems \cite{Garay:1992ej, Deruelle:2009pu}. However, when the lapse and shift are promoted to canonical variables and the extended Dirac phase space $\Gamma_{\rm J}^{\rm (ext)}$ [cf.~Eq.~(\ref{extended-phase-space})] is considered, the conformal transformation no longer preserves the extended symplectic structure \cite{GabrieleGionti:2020drq}. This clarifies that the different conclusions in the literature arise from different choices of canonical phase space rather than from a physical inconsistency between the two frames.

The central result of this work concerns the additional structure introduced by relational dynamics. Once an internal clock is selected and the Hamiltonian constraint is solved for its conjugate momentum, the resulting reduced Hamiltonian depends on the canonical realization of the clock sector. Consequently, canonical equivalence of the unreduced constrained systems does not automatically imply equivalence after relational reduction.

For Brans--Dicke theory, this phenomenon originates from the (canonical conformal) transformation of the scalar momentum:
\begin{equation}
\mathcal{C}:\; p_\phi\; \mapsto\; \tilde p_\phi
=
p_\phi-\frac{p_h}{\phi},
\end{equation}
so that the transformed  momentum mixes with the trace of the gravitational momentum. This mixing reflects the redistribution of dynamical information between the scalar and geometric sectors induced by the conformal transformation. Therefore, although the scalar configuration variable itself is preserved, $\tilde{\phi}=\phi$, the corresponding canonical clock pair is not the same in the two descriptions, $p_\phi\neq \tilde{p}_\phi$.
Using the same configuration variable with different conjugate momenta amounts
to employing different relational clocks and naturally leads to different
reduced Hamiltonians.

This observation leads to the main conceptual conclusion of the present work:
{\em a relational clock is fundamentally the canonical pair $(T, P_T)$, not merely the clock variable $T$}.
Because the physical Hamiltonian is derived directly from $P_T$, relational dynamics are governed by the canonical clock sector. The apparent frame inequivalence between the descriptions originates from a mismatch in the canonical embedding
of this clock sector before reduction.
Equivalently, the phenomenon can be expressed as the non-commutativity between
canonical transformations $\mathcal{C}$ and relational deparametrization $\mathcal R$;
\begin{equation}
\mathcal R\circ\mathcal C
\neq
\mathcal C\circ\mathcal R.
\nonumber
\end{equation}

The apparent frame dependence is resolved by introducing the adapted canonical chart
\begin{equation}
\Gamma_\ast
=
\{Q_{ab},P^{ab};\Phi,P_\Phi\},
\nonumber
\end{equation}
constructed entirely within the original Jordan-frame phase space. The adapted variables are defined by $Q_{ab}=\phi h_{ab}$ and $P_\Phi=p_\phi - p_h/\phi$, 
so that the canonical clock sector $(\phi,p_\phi)$ is consistently replaced by $(\Phi,P_\Phi)$, while preserving the symplectic structure. By construction, $P_\Phi=\tilde p_\phi$, so that deparameterization with respect to the adapted clock momentum yields
\begin{align}
H_{\rm phys}^{\ast}
=
\tilde H_{\rm phys}.
\nonumber
\end{align}
The Jordan and Einstein-frame reduced Hamiltonians therefore generate identical relational dynamics. Frame equivalence is thus restored without modifying the underlying conformal transformation; it is recovered simply by identifying the invariant canonical clock sector before performing a relational deparametrization.

The status of Jordan--Einstein frame equivalence can therefore be summarized as follows:

\begin{table}[ht]
\centering
{\small\begin{tabular}{@{}lll@{}}

Classical action   \quad   & $\Rightarrow$ &\quad  Equivalent \\[3pt]

ADM phase space  \quad   & $\Rightarrow$ & \quad  Canonically equivalent  \\[3pt]

Extended Dirac phase space  \quad   & $\Rightarrow$ & \quad  Generally non-canonical \\[3pt]

Reduced phase space  with naive clock $(\phi, p_\phi)$ \quad   & $\Rightarrow$  & \quad  Clock-sector dependent \\[3pt]

Reduced phase space with adapted clock $(\Phi, P_\Phi)$ \quad   & $\Rightarrow$  & \quad  Canonically equivalent  \\
\end{tabular}}
\end{table}

\vspace{1mm}

The same observation has important consequences for quantum gravity. In a Dirac
quantization approach, where physical states satisfy
\begin{equation}
\hat{\mathcal H}\Psi=0,
\nonumber
\end{equation}
the canonical equivalence of the unreduced constrained systems is naturally preserved. By contrast, reduced phase-space quantization promotes the physical Hamiltonian into the generator of quantum evolution,
\begin{equation}
i\frac{\partial\Psi}{\partial T}
=
\hat H_{\rm phys}\Psi .
\nonumber
\end{equation}
Since this Hamiltonian depends on the canonical realization of the clock,
different choices of clock sectors may lead to different quantum descriptions.
This connects the present analysis with the multiple-choice problem of time in
quantum gravity~\cite{Kuchar:1991qf,Isham:1992ms,Anderson:2012vk}. Whether such differences represent genuine quantum inequivalence or different relational descriptions remains an open question involving the operational treatment of quantum clocks and their backreaction.

The Brans--Dicke theory provides a concrete realization of a more general principle that relational dynamics should be formulated in terms of invariant canonical clock sectors rather than clock configuration variables. This principle is independent of Brans--Dicke theory and may provide a useful framework for studying canonical transformations, relational observables, and reduced phase-space quantization in generally covariant systems.

\section*{Acknowledgments}

Y.T. gratefully acknowledges financial support from the University of Warsaw under the Excellence Initiative -- Research University (IDUB) program (Grant No. PSP: 501-D111-20-0023110), where part of this work was carried out, and from the National Elites Foundation of Iran (BMN) through the Kazemi-Ashtiani grant. 
This work was also conducted within the framework of the following COST (European Cooperation in Science and Technology) Actions: CA18108 ``Quantum Gravity Phenomenology in the Multi-Messenger Approach,'' CA23130 ``Bridging High and Low Energies in Search of Quantum Gravity,'' and CA23115 ``Relativistic Quantum Information.''


\bibliography{main}

\end{document}